\documentclass[nofootinbib,onecolumn,showpacs,preprintnumbers,amssymb,amsmath, amsfonts,psamsfonts]{revtex4}

\usepackage{amssymb,amsmath, amsfonts}
\usepackage[mathscr]{eucal}

\usepackage{graphicx}
\usepackage{dcolumn}
\usepackage{bm}

\newcommand{\tr}{\textrm{tr}}

\newtheorem{lem}{Lemma}[section]
\newtheorem{proof}{Proof}
\newtheorem{defn}{Definition}[section]
\newtheorem{theorem}{Theorem}
\newtheorem{remark}{Remark}
\newtheorem{discussion}{Discussion}

\begin{document}


\title{Quantum Data Compression and Relative Entropy Revisited}%

\author{Alexei~Kaltchenko}
\affiliation{Department of Physics and Computing, Wilfrid Laurier University, Waterloo, Ontario, N2L3C5, Canada.}
\email{akaltchenko@wlu.ca}


\begin{abstract}
B.~Schumacher and M.~Westmoreland have established a quantum analog of a well-known classical information theory result on a role of relative entropy as a measure of non-optimality in (classical) data compression. In this paper, we provide an alternative, simple and constructive proof of this result by constructing quantum compression codes (schemes) from classical data compression codes. %
Moreover, as the quantum data compression/coding task can be effectively reduced to a (quasi-)classical one, we show that relevant results from classical information theory and data compression become applicable and therefore can be extended to the quantum domain.
\end{abstract}


\maketitle

\section{Introduction}
Quantum relative entropy, defined as~$S(\rho ||\sigma)
\triangleq\tr (\rho\log\rho) - \tr (\rho\log\sigma)$, is always
non-negative and is equal to zero if and only if $\rho = \sigma$, where~$\rho$ and~$\sigma$ are density operators defined on the same finite-dimensional Hilbert space. For a review of the properties of quantum relative entropy as well as the properties of quantum (von Neumann) entropy, defined as~$S(\rho) = \tr (\rho \log \rho)$, see, for instance,~\cite{vedral}.

Both quantum (von Neumann) entropy and quantum relative entropy can be viewed as quantum extensions of their respective classical counter-parts and inherit many of their properties and applications. Such inheritance however  is often neither  automatic, nor obvious. For example, classical information theory\cite[Chap.~5]{cover} provides the following natural interpretations of Shannon entropy and relative entropy. A compression code (scheme), for example, Huffman code, which is made optimal for an  information source with a given probability distribution~$p$ of characters, would require $H(p)$ bits per input letter to encode a sequence emitted by the source, where $H(p)$ stands for the Shannon entropy  of  the probability distribution~$p$ and is defined by
\begin{equation}
H(p)\triangleq -\sum\limits {p (x)\log p(x)}.
\end{equation}




 Thus, if using a compression code\footnote{See Subsection~\ref{subsec-optimal} for the formal definitions of information sources, compression codes, etc.} (scheme) made optimal for an information source with probability distribution~$q$, also called  as ``$q$-source'', one could compress such a source with $H(q)$ bits per input character. If, to compress the same~$q$-source, one uses  a {\em different} compression code (scheme)  optimal for a source with a {\em different} probability distribution~$p$ (i.e. $p$-source), it will require $H(q)+D(q ||p)$ bits per input character, where~$D(q||p)$ stands for the relative entropy function for  the probability distribution~$q$ and~$p$  and is defined by
\begin{equation}
D(q||p) \triangleq \sum\limits {q(x)\log
\frac{{q(x)}}{{p(x)}}}.
\end{equation}
Accordingly, the quantity~$D(q ||p)$ represents an additional encoding cost which arises from encoding a $q$-source
with a code optimal for a $p$-source. 
This (classical) data compression interpretation of (classical) Shannon and relative entropies extends\cite{schum,SW} to the quantum domain if one replaces classical information sources having probability distributions~$p$ and~$q$ with quantum information sources having density matrices~$\rho$ and~$\sigma$, respectively, as well as (classical) Shannon and relative entropies  with (quantum) von Neumann and  relative entropies, respectively.

In this paper, we provide an alternative, simple and constructive proof of this result. We build quantum compression codes (schemes) from classical compression codes. 
As the original source density matrix~$\rho$  can simply be replaced with~$\tilde \rho$,  a so-called ``effective'' source matrix\cite{universal-iid-unknown-entropy} with respect to our computational basis, our quantum compression/coding task is effectively reduced to a (quasi-)classical one, with  classical information sources having probability distributions~$p_{\tilde \rho}$ and~$p_{\sigma}$, which are formed by the eigenvalues of~$\tilde \rho$ and~$\sigma$, respectively. Thus, relevant results from classical information theory and data compression become applicable and therefore can be extended to the quantum domain.

Our paper is organized as follows. In Section~\ref{ClassicalDataCompression}, we briefly review some classical data compression definitions and results, including information sources and source codes (also known as ``data compression codes''). In Section~\ref{turning}, we introduce quantum information sources, quantum data compression and  discuss how classical data compression codes can be ``turned'' into quantum data compression codes. In Section~\ref{effective-sec}, we introduce an ''effective'' density matrix for a quantum information source. In Section~\ref{MainResult}, we state our main results: Lemma~\ref{main-lem} and Theorem~\ref{matched-basis} followed by a discussion. We conclude the paper with a brief conclusion section.


\section{Classical Data Compression}\label{ClassicalDataCompression}
\subsection{Classical Information Sources, Compression Codes, and Optimal Data Compression}\label{subsec-optimal}
Let $\mathcal{A}= \left\{ a_1,a_2, \dots,a_{|\mathcal{A}|}
\right\}$ be a  data sequence alphabet, where the notation
$|\mathcal{A}|$ stands for the cardinality of~$\mathcal{A}$.
Binary alphabet~$ \{0, 1 \}$ and Latin alphabet ~$ \{a, b, c, d, \ldots, x, y, z \}$ are widely-used alphabet examples. We will also use the notation $|sequence\_name|$ to denote the
sequence length. Let $\mathcal{A}^*$, $\mathcal{A}^{+}$, and
$\mathcal{A}^{\infty}$ be respectively the set of all finite
sequences, including the empty sequence,  the set of all finite
sequences of positive length, and the set of all infinite
sequences, where all the sequences are drawn from~$\mathcal{A}$.
For any positive integer~$n$, $\mathcal{A}^n$ denotes the set of
all sequences of length $n$ from $\mathcal{A}$. A sequence
from~$\mathcal{A}$ is sometimes called an $\mathcal{A}$-sequence.
To save space, we denote an $\mathcal{A}$-sequence $x_1,
x_2,\ldots, x_n$ by~$x^n$ for every integer $n>0$.
\begin{defn}
An alphabet~$\mathcal{A}$ independently and identically distributed (i.i.d.) information source is an i.i.d. discreet  stochastic process, which is defined by a single-letter (marginal) probability distribution of characters from~$\mathcal{A}$. We extend the single-letter probability distribution to a probability measure on the probability space constructed for the set~$\mathcal{A}^+$ of all $\mathcal{A}$-sequences. We say ``$\mathcal{A}$-source'' or ``$p$-source'' to emphasize that the source has a certain alphabet~$\mathcal{A}$ or a probability distribution~$p$.
\end{defn}
Thus, we assume that every $\mathcal{A}$-sequence is generated by an i.i.d. information source.
\begin{defn}
An \textit{alphabet $\mathcal{A}$ source code}~$\mathbf{C}$ is a
mapping from $\mathcal{A}^+$ into $\{0,1 \}^+$, where, according
to our notation, $\mathcal{A}^+$ is the set of all
$\mathcal{A}$-strings and  $\{0,1 \}^+$ is the set of all binary
strings.
\end{defn}
We restrict our attention to {\em uniquely decidable} codes that is, for every~$n$ and every $x_1,
x_2,\ldots, x_n$, we have
$$ \mathbf{C^{-1}} \left( {\mathbf{C} (x_1, x_2,\ldots, x_n)} \right) = x_1,
x_2,\ldots, x_n$$

\smallskip

\noindent For every~$n$ and every $x_1,x_2,\ldots, x_n$, we call a sequence~$\mathbf{C} (x_1, x_2,\ldots, x_n)$ a {\em codeword}. We often call a source code a {\em compression} code and use the terms ``source code'' and ``compression code'' interchangeably.

\begin{defn}
For any given source  and an alphabet $\mathcal{A}$ source code $\mathbf{C}$, we define the compression rate of~$\mathbf{C}$ in bits per input
character by $\left\langle {\frac{1}{n} \left| \mathbf{C}(x_1, x_2,\ldots, x_n)
\right|} \right\rangle _p$,
where the  length of a compressed sequence (codeword)~$\mathbf{C} (x_1, x_2,\ldots, x_n)$ is divided by~$n$ and averaged with respect to information source's  probability measure~$p$.
\end{defn}

\begin{theorem}[Optimal Data Compression]
For an i.i.d. information source with probability measure~$p$ and any compression code, the best achievable compression rate is given by Shannon entropy~$H(p)$, which is defined by
\begin{equation}
H(p)\triangleq -\sum\limits_{x \in \mathcal{A} } {p (x)\log p(x)}
\end{equation}
\end{theorem}

\begin{defn}
For any probability distribution~$p$, we call a source code {\em optimal} for an i.i.d. source with probability distribution~$p$ if the compression rate for this code is equal to Shannon entropy~$H(p)$. We denote such a code by~$\mathbf{C}_p$.
\end{defn}
\begin{remark}
Optimal compression code(s)  exist for any and every source's probability distribution.
\end{remark}

\subsection{Classical Data Compression and Relative Entropy}

Suppose we have a compression code~$\mathbf{C}_p$, that is~$\mathbf{C}_p$  is  optimal for an $\mathcal{A}$-source with a probability distribution~$p$. What happens to the compression rate if we use~$\mathbf{C}_p$ to compress another $\mathcal{A}$-source with a different distribution~$q$? It is not difficult to infer  that the compression rate will generally increase, but by how much? The increase in the compression rate will be equal to relative entropy function~$D(q||p)$, which is defined by
\begin{equation}
D(q||p) \triangleq \sum\limits_{x\in \mathcal{A} } {q(x)\log
\frac{{q(x)}}{{p(x)}}},
\end{equation}
where relative entropy function~$D(q||p)$ is always non-negaitive\cite{cover} and equal to zero if and only if $q \equiv p$.
More precisely, if using
a code~$\mathbf{C}_q$ (i.e. optimal for a distribution~$q$), one could encode a $q$-source with $H(q)$ bits per input character. If, instead, one uses a code~$\mathbf{C}_p$ (i.e. optimal for~$p$) to encode a $q$-source, it will require $H(q)+D(q ||p )$ bits per input character.

\section{Turning Classical Source Codes into Quantum Transformations}\label{turning}
In this section we  review how a classical one-to-one mapping
(code) can be naturally extended\footnote{For more details, see~\cite{universal-iid-unknown-entropy} or \cite{langford}} to
a unitary quantum transformation. But first, we need to define quantum information sources.
\subsection{Quantum Information Sources}
 Informally,   a quantum  identically and independently distributed  (i.i.d.) information source is an infinite sequence of identically and independently generated quantum states (say, photons, where each photon can have one of two fixed polarizations). Each such state ``lives'' in its ``own'' copy of the same Hilbert space  and is characterized by its ``own'' copy of the same density operator acting on the Hilbert space. We denote the Hilbert space and the density operator  by~${\cal H}$ and~$\rho$, respectively. Then, the sequence of~$n$ states ``lives'' in the space ${\cal H}^{\otimes n}$ and is characterized by the density operator $\rho ^{\otimes n}$
acting on~${\cal H}^{\otimes n}$. For example, in case of a photon sequence, we have $\dim {\cal H} =2 $. Thus, to define a quantum source, we just need to specify~$\rho$ and~${\cal H}$. In this paper, we restrict our attention to quantum sources defined on  two-dimensional Hilbert spaces. Often, just~$\rho$  is given and~${\cal H}$ is not explicitly specified as all Hilbert spaces of the same dimensionality  are isomorphic.
\subsection{Turning Classical Source Codes into Quantum Transformations}
  Now we ready to explain how a classical source code can be turned into a quantum transformation. For the rest of the paper, we will only be considering {\em binary} classical sources with alphabet~$\mathcal{A} = \{0, 1 \}$ and quantum information sources defined on a 2-dimensional Hilbert space.

   Let~$\mathcal{B} \triangleq \bigl\{|0\rangle,
|1\rangle\bigr\}$ be an arbitrary, but fixed orthonormal basis in the quantum information source space~${\cal H}$.   We call~$\mathcal{B}$ a {\em computational basis}. Then, for a sequence of~$n$ quantum states, the orthonormal basis will have $2^n$ vectors of the form
\begin{equation}\label{tenzor-basis}
|e_1\rangle \otimes |e_2\rangle \otimes |e_3\rangle \otimes \cdots \otimes |e_n\rangle,
\end{equation}
where, for every integer~$1 \leqslant i \leqslant n$, $|e_i\rangle \in \mathcal{B}  \triangleq \bigl\{|0\rangle,
|1\rangle\bigr\}$. To emphasize  the symbolic correspondence between  binary strings of length~$n$ and the basis vectors in~${\cal H}^{\otimes n}$, we can rewrite~\eqref{tenzor-basis} as
\begin{equation}
|e_1 e_2 e_3 \ldots e_n\rangle,
\end{equation}
where, for every integer~$1 \leqslant i \leqslant n$, $e_i \in \{0,1\}$. In other words, if we add the Dirac's  ``ket'' notation at the beginning and the end of any binary string of length~$n$, then we will get a basis vector in~${\cal H}^{\otimes n}$.

Thus, for any computational basis~$\mathcal{B} \in \mathcal{H}$, positive integer~$n$, and  a
classical fully reversible function (mapping) for~$n$-bit binary strings, we can define a unitary transformation on ${\cal H}^{\otimes n}$ as follows.

\noindent Let $ \varphi : \{0,1\}^n \rightarrow \{0,1\}^n$ be a
classical fully reversible function for $n$-bit binary strings. For
every $\varphi$ and $n$,  we define a unitary operator
$U_{\mathcal{B}, \varphi}^n: \mathcal{H}_{in} ={\cal H}^{\otimes n} \rightarrow \mathcal{H}_{out}$ by
the bases vectors mapping $U_{\mathcal{B}, \varphi}^n|e_1\dots e_n\rangle
=|\varphi(e_1\dots e_n)\rangle$, where, for every integer~$1 \leqslant i \leqslant n$, $e_i \in \{0,1\}$ and  $|e_i\rangle \in \mathcal{B}  \triangleq \bigl\{|0\rangle,|1\rangle\bigr\}$. We point out that~$\mathcal{H}_{out}$ denotes an output space which is an isomorphic copy of the input space~$\mathcal{H}_{in} ={\cal H}^{\otimes n}$, and~$\mathcal{H}_{out}$ may sometimes coincide with $\mathcal{H}_{in}$.

\smallskip

\noindent An elegant way of designing quantum circuits to implement~$U_{\mathcal{B}, \varphi}^n$ was  given by Langford in~\cite{langford}, where~$U_{\mathcal{B}, \varphi}^n$ was computed by the following two consecutive operations defined on $\mathcal{H}_{in} \otimes \mathcal{H}_{out}$:
\begin{equation}|e_1\dots e_n\rangle |e_1^{out}\dots e_n^{out}\rangle \longrightarrow
|e_1\dots e_n\rangle |e_1^{out}\dots e_n^{out} \oplus \varphi(e_1\dots
e_n)\rangle\label{trans1}
\end{equation}
\begin{equation}
|e_1\dots e_n\rangle |e_1^{out}\dots e_n^{out}\rangle \longrightarrow |e_1\dots
e_n \oplus \varphi^{-1}(e_1^{out}\dots e_n^{out}) \rangle |e_1^{out}\dots
e_n^{out}\rangle,\label{trans2}
\end{equation}
where $e_i, e_i^{out} \in \{0,1\}$, \ $|e_1^{out}\dots e_n^{out}\rangle \in \mathcal{H}_{out}$, and the notation $\oplus$ stands for bit-wise modulo-2 addition.

\smallskip

\noindent In the above, Relation~\eqref{trans1} maps a basis vector $|e_1  \ldots e_n
\rangle |\underbrace {0 \ldots 0}_n\rangle$ to the basis vector
$|e_1 \ldots e_n \rangle |\varphi(e_1  \ldots e_n )\rangle$.
Relation~\eqref{trans2} maps a basis vector $|e_1 \ldots e_n \rangle
|\varphi(e_1 \ldots e_n )\rangle$ to the basis vector $|\underbrace
{0 \ldots 0}_n\rangle |\varphi(e_1 \ldots e_n )\rangle$.

\smallskip

\noindent Another interesting approach for computing~$U_{\mathcal{B}, \varphi}^n$ was suggested in~\cite{universal-iid-unknown-entropy}, where Jozsa and Presnell used a so-called ``digitalization procedure''.
\begin{remark}
We reiterate  that, given a quantum source's Hilbert space~$\mathcal{H}$, the quantum transformation~$U_{\mathcal{B}, \varphi}^n$ entirely depends on- and determined by the selection of the computational basis~$\mathcal{B} \in \mathcal{H}$ and a classical fully reversible function~$\varphi$ (mapping) for~$n$-bit binary strings. We also point out that as long as $\varphi(\cdot)$ and $\varphi^{-1}(\cdot)$ can be classically
computed in linear time and space, then, $U_{\mathcal{B}, \varphi}^n$ can be computed\cite{bennett,BW,TS} in (log-)linear time and space, too.
\end{remark}
\subsection{Quantum Data Compression Overview}
We wish to compress a sequence of quantum states so that at a later stage the compressed sequence can be de-compressed . We also want that the decompressed sequence will be the same as- or ``close'' to the original sequence. The measure of ``closeness'' is called {\em quantum fidelity}. There are a few different definitions of fidelity. One widely used is
the fidelity between any pure state~$|\psi \rangle$ and mixed state~$\rho $, defined by
$$F( |\psi \rangle \langle \psi|, \rho )  = \langle \psi| \rho |\psi \rangle $$


We now introduce a new notation
$U_{\mathcal{B}, \mathbf{C}}^n$ for a unitary transformation
defined as~$U_{\mathcal{B}, \mathbf{C}}^n := \left.
{U_{\mathcal{B}, \varphi}^n} \right|_{\varphi=\mathbf{C}}$. Thus,
the unitary transformation~$U_{\mathcal{B}, \mathbf{C}}^n$ acts on a product space ${\cal
H}^{\otimes  n} \times {\cal H}^{\otimes  n}$ and can be viewed as a quantum extension of the classical compression code~$\mathbf{C}$. So to call~$U_{\mathcal{B}, \mathbf{C}}^n$, we will use terms ``unitary transformation'' and ``(quantum) compression code'' interchangeably.
\begin{defn}
Given a quantum i.i.d. source with density operator~$\rho$, we define a compression rate of $U_{\mathcal{B}, \mathbf{C}}^n$ by
\begin{equation}
\left\langle {\frac{1}{n} \left| U_{\mathcal{B}, \mathbf{C}}^n \ |z_1 z_2 z_3 \ldots z_n\rangle
\right|} \right\rangle,
\end{equation}
where~$|z_1 z_2 z_3 \ldots z_n\rangle \in {\cal H}^{\otimes  n}$  is a sequence emitted by the quantum source, and
the  length of compressed quantum sequence~$U_{\mathcal{B}, \mathbf{C}}^n \ |z_1 z_2 z_3 \ldots z_n\rangle$ is divided by~$n$ and averaged in the sense of a quantum-mechanical observable\footnote{Of course, the length a quantum sequence is a quantum-mechanical variable. See~\cite{SW} for a detailed discussion on the subject}.
\end{defn}
If we ``cut'' the compressed sequence~$U_{\mathcal{B}, \mathbf{C}}^n \ |z_1 z_2 z_3 \ldots z_n\rangle$ at the length ``slightly more'' than~$\left\langle { \left| U_{\mathcal{B}, \mathbf{C}}^n \ |z_1 z_2 z_3 \ldots z_n\rangle
\right|} \right\rangle$, then it can be de-compressed with high fidelity. Thus,~$\left\langle { \left| U_{\mathcal{B}, \mathbf{C}}^n \ |z_1 z_2 z_3 \ldots z_n\rangle \right|} \right\rangle$ is the average number of qubits needed to faithfully represent an uncompressed sequence of length~$n$. See~\cite{SW} for a detailed treatment.

\begin{theorem}[Schumacher's Optimal Quantum Data Compression\cite{schum}]
For a quantum i.i.d. source with a density operator~$\rho$ and any compression scheme, the best achievable compression rate with a high expected fidelity is given by  the von Neumann entropy, which is defined by
\begin{equation}
S(\rho) = \tr (\rho \log_2 \rho)
\end{equation}
\end{theorem}

\begin{defn}
A compression code, for which the above rate is achieved, is called {\em optimal}.
\end{defn}

\begin{lem}
For a density operator~$\rho$, we choose  the eigenbasis of~$\rho$ to be our computational basis~$\mathcal{B}$. Let~$\mathbf{C}_{p_\rho}$ be a classical compression code, which is optimal for a probability distribution formed my the set of the eigenvalues  of~$\rho$. Then, a unitary transformation~$U_{\mathcal{B}, \mathbf{C}_{p_\rho}}^n$ is an optimal quantum compression code for a quantum information source with the density operator~$\rho$.
\end{lem}

\section{``Effective'' source density matrix}\label{effective-sec}
We recall a quantum (binary) identically and independently distributed  (i.i.d.) information source is defined by a density operator acting on a two-dimensional  Hilbert  space.

Let~$\mathcal{B} = \left\{ |b_i\rangle \right\}$ be our computational basis and suppose we have a quantum i.i.d. source with density operator~$\rho$ acting on our two-dimensional computational Hilbert  space. Let $\sum\limits_i {\lambda _i |\lambda _i \rangle \langle \lambda _i
|}$  be an orthonormal decompositions of~$\rho$.

For every $\rho$ and $\left\{ |b_i\rangle \right\}$, Jozsa and
Presnell defined\cite{universal-iid-unknown-entropy} a so-called
``effective'' source density matrix $\tilde \rho$ with respect to
a (computational) orthonormal basis $\left\{ |b_i\rangle \right\}$ as follows
\[
\tilde \rho  \triangleq \sum\limits_j {\eta _j |b_j \rangle \langle
b_j |},
\]
where~$\eta _j  \triangleq \sum\limits_i {P_{ij} \lambda _i }$  and~$\left[ {P_{ij} } \right]$ is a doubly stochastic matrix defined by
$$P_{ij}  = \langle \lambda _i |b_j \rangle \ \langle b_j | \lambda _i \rangle \geqslant 0.$$

We observe the following obvious fact. If the source eigenbasis coincides  with the computation basis, i.e. $\left\{ {|\lambda _i \rangle } \right\} = \left\{ |b_i\rangle \right\}$, then effective source density matrix coincides with the actual source density matrix, i.e. $\tilde \rho  = \rho$.

\section{Main Result}\label{MainResult}
%
\begin{defn}
With any density matrix~$\rho$,  we associate a probability distribution formed by the density matrix's eigenvalues and denote it by putting the density matrix notation as a subscript:~$p_{\rho}$. That is $p_{\rho} = \bigl\{ \lambda_0, \lambda_1  \bigr\}$, where~$\bigl\{ \lambda_0, \lambda_1  \bigr\}$ are the eigenvalues of~$\rho$.
\end{defn}
\begin{lem}\label{main-lem}
Let~$\rho$ and~$\sigma$ be density matrices defined on the same Hilbert space, with
orthonormal decompositions~$\rho = \sum\limits_i {\lambda _i \ |\lambda _i \rangle \langle \lambda _i
|}$ and~$\sigma = \sum\limits_j {\chi _j \ |\chi _j \rangle \langle
\chi _j |}$.

\noindent Let $\tilde \rho$ be the effective density matrix of $\rho$
with respect to the basis $\left\{{|\chi_i \rangle} \right\}$, and
let $\left\{{\eta _j } \right\}$ be the eigenvalues of $\tilde
\rho$.

\smallskip

\noindent Then, we have the following relation between quantum and classical entropies:
\begin{equation}\label{S-D-H-S-1}
S(\rho ||\sigma ) +  S(\rho )= D(p_{\tilde \rho} ||p_{\sigma} ) + H(p_{\tilde
\rho} ),
\end{equation}
where~$D(\cdot ||\cdot )$ is classical  relative entropy, $H(\cdot)$  is classical (Shannon) entropy, $p_{\tilde \rho}$ and $p_{\sigma}$ are  the probability distributions formed by eigenvalues sets of~$ \tilde \rho$ and~$\sigma$, respectively. That is we have $p_{\tilde
\rho}  = \left\{{\eta _j } \right\}$ and
$p_{ \sigma}  = \left\{{\chi_i } \right\}$.
\end{lem}
\begin{proof}
By the definition of  effective density matrix, we have
$$
\eta _j  \triangleq \sum\limits_i {P_{ij} \lambda _i },
$$
where $\left[ {P_{ij} } \right]$ is a doubly stochastic matrix,
defined by
\[
P_{ij}  = \langle \lambda _i |\chi _j \rangle \ \langle \chi _j | \lambda _i \rangle.
\]
By the definition of  quantum relative entropy, we have
\begin{equation*}
S(\rho ||\sigma ) \triangleq \tr (\rho \log  \rho) - \tr (\rho \log \sigma).
\end{equation*}
First, we look at the quantity~$\tr (\rho \log \sigma)$:
\begin{equation*}
\tr (\rho \log \sigma) = \tr (\rho \log \sigma  \sum\limits_i { |\lambda _i \rangle \langle \lambda _i
|})  = \sum\limits_i { \tr (\rho \log \sigma   |\lambda _i \rangle \langle \lambda _i|) }
= \sum\limits_i {\langle \lambda _i|  \rho \log \sigma   |\lambda _i \rangle }.
\end{equation*}
Substituting the identity~$\langle \lambda _i|  \rho = \lambda _i  \langle \lambda _i| $ into the above equation, we get
\begin{equation}\label{rho-log-sigma}
\tr (\rho \log \sigma)
= \sum\limits_i {\lambda _i \ \langle \lambda _i|  \log \sigma   |\lambda _i \rangle   }
\end{equation}
On the other hand, we have the following identity for the $\log$ function:
\begin{equation}\label{log-function}
 \log \sigma  = \sum\limits_j { \log (\chi _j) \ |\chi _j \rangle \langle
\chi _j |}
\end{equation}
and, therefore, substituting the right-hand-side of~\eqref{log-function} into $\langle \lambda _i|  \log \sigma   |\lambda _i \rangle $, we get
\begin{equation}\label{log-sigma}
\langle \lambda _i|  \log \sigma   |\lambda _i \rangle   = \langle \lambda _i|  \left({ \sum\limits_j { \log (\chi _j) \ |\chi _j \rangle \langle
\chi _j |}  }\right)  |\lambda _i \rangle = \sum\limits_j { \log (\chi _j) \ \langle \lambda _i |\chi _j \rangle \ \langle \chi _j | \lambda _i \rangle } = \sum\limits_j { \log (\chi _j) \ P_{ij} }
\end{equation}
Combining equations~\eqref{rho-log-sigma} and~\eqref{log-sigma}, we obtain
\begin{equation}\label{rho-log-sigma1}
\tr (\rho \log \sigma)
= \sum\limits_i {\lambda _i \  \sum\limits_j { \log (\chi _j) \ P_{ij} }  } =  \sum\limits_j { \log (\chi _j)  \  \sum\limits_i {  P_{ij} \lambda _i }  } = \sum\limits_j {\eta_j \ \log (\chi _j)}.
\end{equation}

\noindent Substituting the right-hand-side of the quantum entropy  identity~$S(\rho ||\sigma ) = \sum\limits_i {\lambda _i \log \lambda _i}$ and that of~\eqref{rho-log-sigma1} into the definition of  quantum relative entropy, we have
\begin{equation}
S(\rho ||\sigma ) = \sum\limits_i {\lambda _i \log \lambda _i} - \sum\limits_j {\eta_j \ \log \chi _j} = - S(\rho
) - \sum\limits_j {\eta_j \ \log \chi _j} = - S(\rho
) - \sum\limits_j {\eta_j \ \log \chi _j} + \sum\limits_j {\eta_j \ \log \eta_j}  - \sum\limits_j {\eta_j \ \log \eta_j}
\end{equation}
Thus, we get
\begin{equation}\label{S-D-H-S}
S(\rho ||\sigma ) = - S(\rho) + D(\eta _j ||\chi _j ) + H(\eta _j ),
\end{equation}
where $D(\eta _j ||\chi _j )$ is the (classical) relative entropy
of the probability distributions~$\eta _j$ and~$\chi
_j$; $H(\eta _j$ ) is the (classical) entropy of the distribution
$\eta _j$.
\smallskip
\noindent We can rewrite~\eqref{S-D-H-S} as~\eqref{S-D-H-S-1} which completes the proof.
\end{proof}

In Subsection~\ref{subsec-optimal}, we have introduced a notation $\mathbf{C}_p$ to  denote a classical data compression code (mapping) optimal for a source probability distribution~$p$. Then, according to our notation, the code~$\mathbf{C}_{p_\sigma}$ will be optimal for a (classical) i.i.d. source with probability distribution  formed by the eigenvalues of a density matrix~$\sigma$.
%
%
\begin{lem}[Classical data compression and relative entropy]\label{classical-rel-entr}
Let~$z^n \triangleq z_1,\ldots,z_n$ be a sequence emitted by a classical i.i.d. source with a marginal probability distribution formed by the eigenvalues of~$ \tilde \rho$. Then, from classical data compression theory, we have
\begin{equation}
\left\langle {\frac{1}{n} \left| \mathbf{C}_{p_\sigma}(z^n)
\right|} \right\rangle _{p_{\tilde \rho}} =
D(p_{\tilde \rho} ||p_{\sigma} ) + H(p_{\tilde \rho}),
\end{equation}
where the notation in the left-hand-side stands for the  length of compressed sequence (codeword) divided by~$n$ and averaged with respect to the probability measure~$p_{\tilde \rho}$.
\end{lem}
%
Combining Lemma~\ref{main-lem}, Lemma~\ref{classical-rel-entr}, and the results of Section~\ref{turning}, we obtain the following theorem:
\begin{theorem}[Quantum Data Compression and Relative Entropy]\label{matched-basis}
Let $|z_1 z_2 z_3 \ldots z_n\rangle \in {\cal H}^{\otimes  n}$   be emitted by a quantum i.i.d. source
with density operator~$\rho$ defined on a Hilbert space~${\cal H}$. Let~$\sigma$ be an arbitrary density operator defined on the same Hilbert space~${\cal H}$.

\smallskip

\noindent We choose  the eigenbasis of~$\sigma$ to be our computational basis~$\mathcal{B}$,
 and let~$\tilde \rho$ be the effective matrix of~$\rho$
with respect to~$\mathcal{B}$. Then, the following relations hold:
\begin{equation}\label{estimate}
S(\rho ||\sigma ) +  S(\rho ) =
D(p_{\tilde \rho} ||p_{\sigma} ) + H(p_{\tilde \rho}) = \left\langle {\frac{1}{n} \left| U_{\mathcal{B}, \mathbf{C}_{p_\sigma}}^n \ |z_1 z_2 z_3 \ldots z_n\rangle
\right|} \right\rangle,
\end{equation}
where the notation in the right-hand-side stands for the  length of compressed quantum sequence divided by~$n$ and averaged in the sense of a quantum-mechanical observable.
\end{theorem}
\begin{remark}
 Quantum relative entropy~$S(\rho ||\sigma )$ has the following operational meaning. If a quantum compression code, optimal for one quantum information source with a density operator~$\sigma$, is applied to another source  with a density operator~$\rho$, then the increase  of the compression  rate is equal to~$S(\rho ||\sigma )$.
\end{remark}
\begin{remark}
By fixing~$\rho$ and the computational basis~$\mathcal{B}$ (i.e. the eigenbasis of~$\sigma$), we will also fix$~\tilde \rho$. Then, for a fixed~$~\tilde \rho$, $D(p_{\tilde \rho} ||p_{\sigma} )$ is minimized (and is therefore  equal to zero) when $\sigma = \tilde \rho$. Then, as it follows from Theorem~\ref{matched-basis}, for an arbitrary computational basis~$\mathcal{B}$, the best possible compression rate  is equal to~$H(p_{\tilde \rho}) = S(\tilde \rho)$ and is achieved with transformation~$U_{\mathcal{B}, \mathbf{C}_{\tilde p}}^n$. In~\cite{universal-iid-unknown-entropy}, such setup is called   ``mismatched  bases''  compression.

If we choose  the eigenbasis of~$\rho$ to be our computational basis~$\mathcal{B}$, then we have~$\tilde \rho = \rho$, and optimal quantum compression is achieved with~$U_{\mathcal{B}, \mathbf{C}_{p_\rho}}^n$, with compression rate $H(p_{\rho}) = S(\rho)$.
\end{remark}
\begin{discussion}
 As long as~$U_{\mathcal{B}, \mathbf{C}_{p_\sigma}}^n$ is computed in a computational basis~$\mathcal{B}$, we can simply replace the original source density matrix~$\rho$ with~$\tilde \rho$, the effective one with respect to the computational basis. Then, our quantum compression/coding task is effectively reduced to a classical one, with classical information sources having probability distributions~$p_{\tilde \rho}$ and~$p_{\sigma}$. Therefore, all relevant results from classical information theory  and data compression become applicable. We have already seen one such result, the statement of data compression (non-)optimality~$D(p_{\tilde \rho} ||p_{\sigma} ) + H(p_{\tilde \rho})$.

 To give another example, we consider the following task of quantum relative entropy estimation.  Suppose we have two quantum i.i.d. sources with unknown density matrices~$\rho$ and~$\sigma$, where~$\rho$ and~$\sigma$ defined on the same Hilbert space. We want to estimate quantum relative entropy~$S(\rho ||\sigma )$.

 Let~$|z_1 z_2 z_3 \ldots z_n\rangle$  and~$|x_1 x_2 x_3 \ldots x_n\rangle $ be emitted by the sources with density matrices~$\rho$ and~$\sigma$, respectively. From Theorem~\ref{matched-basis}, one could (hope) to estimate the quantity~$S(\rho ||\sigma ) +  S(\rho )$ by simply measuring the average  length  of codeword~$U_{\mathcal{B}, \mathbf{C}_{p_\sigma}}^n \ |z_1 z_2 z_3 \ldots z_n\rangle$, then estimate~$S(\rho)$ by measuring the average  length  of codeword~$U_{\mathcal{B}, \mathbf{C}_{p_\rho}}^n \ |z_1 z_2 z_3 \ldots z_n\rangle$, and then subtract the latter  from the former. For such straightforward  approach to work, one should have quantum compression codes~$U_{\mathcal{B}, \mathbf{C}_{p_\sigma}}^n$ and~$U_{\mathcal{B}, \mathbf{C}_{p_\rho}}^n$, which are optimal for sources with density matrices~$\rho$ and~$\sigma$, respectively. But as~$\rho$ and~$\sigma$ are unknown, so are~$U_{\mathcal{B}, \mathbf{C}_{p_\sigma}}^n$ and~$U_{\mathcal{B}, \mathbf{C}_{p_\rho}}^n$. Fortunately, in classical data compression, there exist so-called ``universal compression codes''\cite{relative-ziv}, which will let estimate (classical) quantities~$D(p_{\tilde \rho} ||p_{\sigma} ) + H(p_{\tilde \rho})$ and~$ H(p_{\tilde \rho})$ even without knowing probability  distributions~$p_{\tilde \rho}$ and~$p_{\sigma}$. So, based on these classical codes, we can construct quantum codes\cite{kaltchenko} (unitary transformations) to estimate~$S(\rho ||\sigma )$.

\end{discussion}

\section{Conclusion}
We have provided a simple and constructive proof of a quantum analog of a well-known classical information theory result on a role of relative entropy as a measure of non-optimality in data compression. We have constructed  quantum compression codes (schemes) from classical data compression codes. %
Moreover, as the quantum data compression/coding task can be effectively reduced to a (quasi-)classical one, we show that relevant results from classical information theory and data compression become applicable and therefore can be extended to the quantum domain.

\end{document}